\documentclass{article}
\usepackage{tikz}
\usepackage{amssymb}
\usepackage{amsmath}
\usepackage{amsthm}
\usepackage{pdfpages}
\usepackage{verbatim}
\usepackage{caption}
\usepackage{listings}
\usepackage{subcaption}
\usepackage{qtree}
\usepackage[all]{xy}
\usepackage{algorithm}
\usepackage{algpseudocode}
\usepackage{algorithmicx}
\usepackage[colorlinks, plainpages,urlcolor=blue]{hyperref}

\usepackage{url}

\newtheorem{theorem}{Theorem}
\newtheorem{corollary}{Corollary}
\newtheorem{lemma}{Lemma}

\theoremstyle{definition}

\newtheorem{definition}{Definition}

\newtheorem{property}{Property}
\newtheorem{notation}{Notation}

\usepackage[hmarginratio=1:1,top=32mm,columnsep=20pt]{geometry} 

\title{\vspace{-15mm}\fontsize{24pt}{10pt}\selectfont\textbf{Perfect Memory Context Trees in time series modeling}} 

\author{
\large
\textsc{Tong Zhang}\\[2mm] 
\normalsize Department of Mathematics, Northeastern University \\ 
\normalsize {zhang.tong@husky.neu.edu} 
\vspace{-5mm}
}
\date{}

\begin{document}

\maketitle 

\begin{abstract}

\noindent  
The Stochastic Context Tree (SCOT) is a useful tool for studying infinite random sequences generated by an $m$-Markov Chain (m-MC).  It captures the phenomenon that the probability distribution of the next state sometimes depends on less than $m$ of the preceding states. This allows compressing the information needed to describe an m-MC. The SCOT construction has been earlier used under various names: VLMC, VOMC, PST, CTW. 

In this paper we study the possibility of reducing the m-MC to a 1-MC on the leaves of the SCOT. Such context trees are called \emph{perfect-memory}. We give various combinatorial characterizations of perfect-memory context trees and an efficient algorithm to find the minimal perfect-memory extension of a SCOT.

\end{abstract}

\section{Introduction}
Consider a data source with unknown distribution characteristics, which can generate sequences $x_{1}, \, x_{2}, ..., \, x_N$ of letters from an alphabet $A=\{a_{1}, \, \dots, \, a_{n}\}$ in a time homogeneous way. If the distribution of each letter $x_k$ depends on the previous $m$ letters $x_{k-1},x_{k-2},\dots, x_{k-m}$ then the sequence can be described by an $m$-Markov chain.

Storing all the information of an $m$-Markov chain, however, can be costly in practice. Thus it becomes desirable to represent a chain in a more efficient way. In many applications it turns out that one doesn't have to know the full $m$ letter history to accurately describe the distribution of the next letter in the stream. It is also often referred to as a Variable Length Markov Chain (VLMC).

One can define $contexts$ as subsequences of letters in the data stream that contain sufficient information about the distribution of the next letter. Thus, contexts are the memory segments as strings that are required for the accurate description of the chain. For any given chain, the set of contexts can be arranged in a \emph{context tree} structure as follows:
 We say that a string $v$ is a postfix of a string $s$, when there exists a string $u$ such that $s=\overline{uv}$. ($\overline{uv}$ means a concatenation of u and v in the natural ordering.) A set $\mathcal{T}$ of strings (and perhaps also of semi-infinite sequences) is called a \emph{context tree} if no $s_1\in\mathcal{T}$ is a postfix of any other $s_2\in\mathcal{T}$. 
 We draw the context tree as a rooted tree where each vertex other than the root is labeled by a letter from the alphabet $A$. A context is a string (a sequence of letters) that follows the path starting from any leaf and ending at the root. A semi-infinite sequence $a_{-\infty}^{-1}\in\mathcal{T}$ is an infinite path to the root. For example, the context marked with * is 1101 in the context tree below.
 
  \Tree[.Root [.0 [.0 ]
               [.1 [.0 ]
                [.1 ] ]]
                 [.1 [.0 [.0 ]
                 [.1 [.0 ]
                 [.1* ]]]
                [.1 ]]]
            \noindent 
 
 It turns out that context trees where each node is either a leaf or has exactly $n$ children (one for each letter in $A$) are particularly important. Such trees are referred to as \emph{complete} (see Definition~\ref{definition1}).
 
A context tree of maximal depth $m$ can be used to naturally represent an $m$-Markov chain by assigning  to each leaf a probability distribution on the alphabet describing the distribution of the next letter~\cite{IEEEhowto:IP15,IEEEhowto:IP14}. Similarly, a context tree can be used to naturally represent a VLMC. Such a structure is called a Stochastic Context Tree (SCOT).  The root represents the upcoming, yet unknown, future state and each child of a node  represents a possible previous state. See figure below.~\cite{IEEEhowto:IP15}

\begin{center}
\vspace{-1em}
\begin{tikzpicture}[x=0.5cm, y=0.5cm]
\node at (0, 0) (a) {$X_0$(root)}
	child {node at (-4, -1) {$\displaystyle P(X_0=1)=\frac{1}{2}$} edge from parent node [left=10pt] {0}}
	child {node at (4, -1) (b) {$ X_{-1}$}  
		child {node at (-4, -1) (d) {$\displaystyle P(X_0=1)=\frac{1}{4}$} edge from parent node [left=10pt] {0}} 
		child {node at (4, -1) (e) {$\displaystyle P(X_0=1)=\frac{3}{4}$} edge from parent node [right=10pt] {1}} 
		edge from parent node [right=10pt] {1}};

\end{tikzpicture}
\end{center}

  Then, given an $m$-Markov chain, it can always be represented as a SCOT on the complete $n$-ary tree of depth $m$. However, ideally one would like to try to find the most economical  SCOT representation of the chain. 
 
When trying to find the most economical SCOT representation, however, a curious problem may arise. Suppose that the string $s = x_1 x_2 \dots x_k$ is sufficient to determine the distribution of the next state ($x_{k+1}$), but for some $x_{k+1}=a \in A$ the string $sa$ is not enough to predict the subsequent state ($x_{k+2}$). This implies that $s$ should not be used as a context. Despite that it allows prediction of the next state, it doesn't carry sufficient information to predict later states. We say that a context tree is \emph{perfect-memory} tree (or ``has perfect memory'') if such a situation can never arise. In other words the requirement is that when passing from $x_i$ to $x_{i+1}$, the context for the latter should not include older symbols other than those within the context ending with $x_{i+1}$. We formalize this in Definition~\ref{definition4}.

Thus, for the purpose of applications, it is important to recognize whether a context tree has perfect memory. The only known result about this class of context trees has been the following sufficient condition.
                
\begin{theorem}[\cite{IEEEhowto:IP14}]\label{theorem1}
 Let $\mathcal{T}$ be a complete context tree. If there is a $k$ such that  the distance of each leaf from the root is either $k$ or $k+1$, then $\mathcal{T}$ has perfect-memory.
\end{theorem} 

The goal of this paper is twofold. First, we provide a simple characterization of perfect-memory context trees in terms of it's subtrees. For a context tree $\mathcal{T}$ let us denote by $\mathcal{T}_1,\mathcal{T}_2,\dots,\mathcal{T}_n$ the subtrees rooted at the children of $\mathcal{T}$'s root. For two context trees $\mathcal{T}, \mathcal{T}'$ let us write $\mathcal{T}' \subseteq \mathcal{T}$ if $\mathcal{T}'$ is contained in $\mathcal{T}$ such that they have the same root.  Our first main results is the following. 

\begin{theorem}\label{theorem2}
 Let $\mathcal{T}$ be a complete context tree. Then, $\mathcal{T}$ has perfect-memory if and only if   $\ \forall i\in \{1,...,n\}, \  \mathcal{T}_i \subseteq \mathcal{T} $.
\end{theorem} 

In the process we study the partially ordered set of of perfect-memory context trees. We find that this is a lattice (that is, intersections and unions preserve perfect-memory), which in turn allows us to talk about the perfect-memory closure of a context tree.  Our second goal is to give a simple algorithm to construct the perfect-memory closure of an arbitrary context tree.

The paper is organized as follows. In Section~\ref{sec:prelim} we give the detailed definitions and notation, and derive some simple properties of context trees. In Section~\ref{sec:perfect-memory} we formally introduce perfect-memory context trees, establish their main properties, prove Theorem~\ref{theorem2} , and describe some of its corollaries.
In Section~\ref{sec:alg} we describe a linear algorithm for finding the perfect-memory closure. Finally in Section~\ref{sec:examples} we provide examples to illustrate some of our results. 
 
\section{Preliminaries}\label{sec:prelim}

In this section we provide the basic concepts about context trees.

Fix a nonempty and finite set $A=\{a_{1}, \, \dots, \, a_{n}\}$. The set $A$ is called the \textit{alphabet}, and its elements are called \textit{letters}. A finite sequence of elements of $A$, including the empty sequence, is called a \textit{string}. The set of all strings is denoted by $S$. 

For two strings $u,v\in S$, we denote the concatenation of $u$ and $v$ in the natural ordering by $\overline{uv}$. We say that a string $v\in S$ is a \textit{postfix} of a string $s\in S$, denoted by $v\prec s$, if there exists a string $w\in S$ such that $s=\overline{wv}$.\\
In order to not abuse the terminology we give the definitions as follows:
\begin{definition}\label{definition1}
 A \textit{``context tree''} on the alphabet $A$ is a rooted tree where each vertex, except the root, is labeled by one letter in $A$ and no two siblings are labeled by the same letter. A context tree is said to be \textit{``complete''}, if each node is either a leaf or has exactly $|A|$ children. In a context tree, a \textit{``context''} is a string that follows the letters along a path starting from a leaf and ending at the root.
\end{definition}

For a context tree $\mathcal{T}$, we denote the set of all its contexts by $\mathcal{T}^*$. The figure below provides an example of a complete context tree $\mathcal{T}$ over the alphabet $\{0,1\}$:\vspace{0.2cm}

  \Tree[.Root [.0 [.0 ]
               [.1 [.0 ]
                [.1 ] ]]
                 [.1 [.0 [.0 ]
                 [.1 [.0 ]
                 [.1 ]]]
                [.1 ]]]
            \noindent 
            
            In this context tree, the set of contexts is $\mathcal{T}^*$=$\{00,010,110,001,0101,1101,11\}$.\medskip
    
Note that for each context $c$ in a context tree $\mathcal{T}$ there is a unique path $\mathcal{P}_{c}$ in $\mathcal{T}$ starting from a leaf and ending at a root that follows all the letters in $c$. 

It can be easily verified that if $\mathcal{T}$ is a context tree, then no context $s\in\mathcal{T}^*$ is a postfix of any other context $s'\in\mathcal{T}^*$.\medskip
\begin{definition}
\label{definition2}
For nonempty context trees $\mathcal{A}$, $\mathcal{B}$ over the same alphabet, we write $\mathcal{A} \subseteq \mathcal{B}$ to denote \textit{``$\mathcal{A}$ is contained  in $\mathcal{B}$ at the root''}, that is

 \[ \forall \  a \in \mathcal{A}^*,\ \exists \ b\in \mathcal{B}^*:\ a\prec b.\]
\end{definition}
The following lemma gives an equivalent condition, provided that the context trees are complete.             
It will be used to prove Theorem~\ref{theorem2} in the next section. \\

\begin{lemma}\label{lemma1}

  Let $\mathcal{A},\ \mathcal{B}$ be complete context trees over the same alphabet $A$. Then, the following statements are equivalent.
\begin{enumerate}
 \item[(i)]    $\mathcal{A} \subseteq \mathcal{B}$ 
  
 \item[(ii)] $\forall \  b \in \mathcal{B}^*,\exists\  a\in \mathcal{A}^*$:\ \;$a\prec b$.
 \end{enumerate}

\end{lemma}

\begin{proof}
 
 ``(i) $\Longrightarrow$ (ii)":    
 Since $\mathcal{A}$ and $\mathcal{B}$ are both complete context trees over the same alphabet, $\forall b\in \mathcal{B}^*$, $\exists \  a\in \mathcal{A}^*$ s.t. $a\prec b$ or $b\prec a$. If $a\prec b$, we are done. If $b\prec a$, then because of (i), $\exists b' \in \mathcal{B}^*$, s.t. $a\prec b'$. We have $b\prec a\prec b'$, $b$ and $b'$ are both contexts in $\mathcal{B}^*$, then $b=b'=a$, thus, $a\prec b$.

 ``(ii) $\Longrightarrow$ (i)": 
  Since $\mathcal{A}$ and $\mathcal{B}$ are both complete context trees over the same alphabet, $\forall  \ a\in \mathcal{A}^* $, $\exists \  b\in \mathcal{B}^*$ s.t. $a\prec b$ or $b\prec a$. If $a\prec b$, by Definition \ref{definition2} we are done. If $b\prec a$, then because of (ii), $\exists \ a'\in \mathcal{A}^*$, s.t. $a'\prec b$. We have $a'\prec b\prec a$, $a$ and $a'$ are both contexts in $\mathcal{A}^*$, therefore $a'=a=b$, thus, $a\prec b$, $\mathcal{A} \subseteq \mathcal{B}$. 
  \end{proof}

Note that it is essential in Lemma \ref{lemma1} that $\mathcal{A}$ and $\mathcal{B}$ are both complete. Indeed, consider the context $\mathcal{A}$ and $\mathcal{B}$ in the figure below. \vspace{.2cm}

\makeatletter
\newcommand{\rightDashedLines}[1]{
\begin{picture}(2,.5)
\put(0,0){\line(2,1){1}}
\multiput(2,0)(-.2,.1){5}{\line(-2,1){0.1}}
\end{picture}}
\let\qdrawReal=\qdraw@branches
\newcommand\brOverrideRight{\let\qdraw@branches=\rightDashedLines}
\newcommand{\leftDashedLines}[1]{
\begin{picture}(2,.5)
\multiput(0,0)(.2,.1){5}{\line(2,1){0.1}}
\put(2,0){\line(-2,1){1}}
\end{picture}}
\let\qdrawReal=\qdraw@branches
\newcommand\brOverrideLeft{\let\qdraw@branches=\leftDashedLines}

\newcommand{\rightEmpty}[1]{
\begin{picture}(2,.5)
\put(0,0){\line(2,1){1.1}}
\end{picture}}
\let\qdrawReal=\qdraw@branches
\newcommand\brOverrideRightEmpty{\let\qdraw@branches=\rightEmpty}

\newcommand{\bothDashedLines}[1]{
\begin{picture}(2,.5)
\multiput(0,0)(.18,.09){6}{\line(2,1){0.1}}
\multiput(2,0)(-.18,.09){6}{\line(-2,1){0.1}}
\end{picture}}
\let\qdrawReal=\qdraw@branches
\newcommand\brOverrideBoth{\let\qdraw@branches=\bothDashedLines}
\newcommand\brRestore{\let\qdraw@branches=\qdrawReal}
\makeatother

 \Tree[.Root($\mathcal{A}$) [.0 [.0 ]
               [.1 ]]
                 [.1 [.0 ]
                [.1 ]]]
 \Tree[.Root($\mathcal{B}$) [.0 [.0 ]
               [.1 ]  ]   
                 [.1 [.0 ] [.~ ] !{\brOverrideRightEmpty} ] !{\brRestore} ]
                 
\vspace{.2cm}
                 
We have $\mathcal{A}^*=\{00,10,01,11\}$, $\mathcal{B}^*=\{00,10,01\}$. Therefore,  condition (ii) of Lemma \ref{lemma1} is satisfied. However, $\mathcal{A}$ is not contained in $\mathcal{B}$ at the root, since 11 is not a postfix of any string in $\mathcal{B}^*$. In the other direction, it is clear that $\mathcal{B}$ is contained in $\mathcal{A}$ at the root but 11 does not have any postfix that is a context in $\mathcal{B}$.

Given two complete context trees over the same alphabet, then each context in a context tree is either a postfix of some context in the other context tree or has a postfix that is a context in the other context tree. Based on this fact, we can define the following.
\begin{definition}
\label{definition3}
For nonempty complete context trees $\mathcal{A}$, $\mathcal{B}$ over the same alphabet, we write $\mathcal{A} \cap \mathcal{B}$ to denote \textit{``the intersection at the root'' of $\mathcal{A}$ and $\mathcal{B}$} and we write $\mathcal{A} \cup \mathcal{B}$ to denote \textit{``the union at the root'' of $\mathcal{A}$ and $\mathcal{B}$}, they are
\[(\mathcal{A}\cap \mathcal{B})^*=\{u|u\in \mathcal{A}^*, \ \exists c\in \mathcal{B}^*, \text{ s.t. } u\prec c;\  \text{or }\ u\in \mathcal{B}^*, \ \exists c\in \mathcal{A}^*, \text{ s.t. } u\prec c \};\]
\[(\mathcal{A}\cup \mathcal{B})^*=\{u|u\in \mathcal{A}^*, \ \exists c\in \mathcal{B}^*, \text{ s.t. } c\prec u;\  \text{or }\ u\in \mathcal{B}^*, \ \exists c\in \mathcal{A}^*, \text{ s.t. } c\prec u \}.\]
\end{definition}

Note that $(\mathcal{A}\cap \mathcal{B})^*$ is not equal to $\mathcal{A}^*\cap \mathcal{B}^*$ most of the time. 
 The next lemma shows the relationship between the two expressions. 
\begin{lemma}
\label{property2}
Let $\mathcal{A},\ \mathcal{B}$ be complete context trees on the same alphabet $A$. Then we have\\
$(\mathcal{A}\cap \mathcal{B})^* \cap (\mathcal{A}\cup \mathcal{B})^*= \mathcal{A}^*\cap \mathcal{B}^* $ and 
$(\mathcal{A}\cap \mathcal{B})^* \cup (\mathcal{A}\cup \mathcal{B})^*= \mathcal{A}^*\cup \mathcal{B}^* $.\\
\end{lemma}

\begin{lemma}
There is a natural bijection map between the set of finite context trees and the set of finite complete context trees over the same alphabet.
\end{lemma}
\begin{proof}
For any complete context tree, we can remove all the leaves to get a ``parent tree'' which is an arbitrary context tree. On the opposite, given an arbitrary context tree, by making sure each node grows $n$ leaves we get a complete context tree.  
\end{proof}
The total number of complete context trees on $A$ with length no more than $\ell$ is $\Omega(n^{n^\ell})$.

  \section{Perfect-memory Context Trees} \label{sec:perfect-memory}
  In this section we define perfect-memory context trees and analyze their properties.
\begin{definition}
\label{definition4} A context tree $\mathcal{T}$ on the alphabet $A=\{a_{1}, \, \dots, a_{n}\}$ is called \textit{``perfect-memory"} context tree (or $\mathcal{T}$ has ``perfect-momory'', denoted as $\mathcal{T}\in \mathcal{PM}$) if  $\forall \ c\in \mathcal{T}^*, i \in \{1,...,n\}, \ \exists \ u \in \mathcal{T}^*$, 
 s.t. $u\prec \overline{ca_i}$. \\
\end{definition}
  Not all complete context trees have perfect-memory. The binary context tree in the figure below provides an example for this. 
  
  \Tree[.Root [.0 [.0 ]
               [.1 ]]
                 [.1 [.0 [.0 ]
                 [.1 [.0 ]
                 [.1 ]]]
                [.1 ]]]
  \vspace{0.2cm}              
	
\noindent The tree is obviously complete, and we have $\mathcal{T}^*$=$\{00,10,001,0101,1101,11\}$.  However, for the context $c=10$ in $\mathcal{T}^*$ and letter $a=1$, the concatenation is $\overline{ca}=101$, and there is no context $u\in \mathcal{T}^*$ that is a postfix of $101$. Thus, the tree does not have perfect-memory. 

Nevertheless, the converse is true, and perfect-memory implies completeness.


\begin{property}\label{property1}
A perfect-memory context tree is complete.
\end{property}
\begin{proof}
Assume there exists a perfect-memory context tree $\mathcal{T}$ that is not complete. Let $c$ be a minimum length string that makes $\mathcal{T}$ not complete (a shortest string that $\mathcal{T}$ is missing to be complete), i.e. $c \notin \mathcal{T}^*$, the leaf of path $\mathcal{P}_c$ has at least one sibling in $\mathcal{T}$ and
\begin{quote}
$\forall c'$ s.t. $c\prec c'$, we have $c'\notin \mathcal{T}^*$.\hfill (i)
\end{quote}
Therefore, by the definition of context tree and that the leaf of context path $\mathcal{P}_c$ has a sibling in $\mathcal{T}$, we have 
\begin{quote} 
$\forall c'$ s.t. $c'\prec c $, we have $c' \notin \mathcal{T}^*$\hfill (ii)
\end{quote} 
Let $c=\overline{wa_c}$, i.e. the last letter of c is $a_c$. If $c=a_c$, i.e. $w$ is the empty string, pick some $t\in \mathcal{T}^*$. Because $\mathcal{T}$ has perfect-memory, $\exists u_o \in \mathcal{T}^*$ s.t. $u_o \prec \overline{ta_c}$. But then $c\prec u_o$, which contradicts (i). Therefore, $w$ cannot be the empty string.

If $w \in \mathcal{T}^*$, then by Definition \ref{definition4}, $\exists u\in \mathcal{T}^*$ s.t. $u\prec \overline{wa_c}=c$. But this is a contradiction of (ii), therefore, $w \notin \mathcal{T}^*$.

Since $w$ is shorter than $c$, and using the assumption that $c$ is the shortest string that makes $\mathcal{T}$ not complete, we know that there exists $u' \in \mathcal{T}^*$ s.t. $u'\prec w$ or $w\prec u'$. By Definition \ref{definition4}, $\exists u_1\in \mathcal{T}^*$ s.t. $u_1\prec \overline{u'a_c}$.

If $u'\prec w$, then $u_1\prec \overline{u'a_c}\prec \overline{wa_c}=c$. This contradicts (ii), thus, $w\prec u'$. Therefore, we have $c=\overline{wa_c}\prec \overline{u'a_c}$ and $u_1\prec \overline{u'a_c}$. $u_1$ and $c$ are both postfixes of $\overline{u'a_c}$, then one must be the postfix of the other. If $u_1\prec c$, contradiction of (ii). If $c\prec u_1$, this is a contradiction of (i).
Therefore, the assumption always leads to a contradiction, which means that a perfect-memory context tree has to be complete. 
\end{proof}

\begin{property}\label{property2}
 If $\mathcal{A}$, $\mathcal{B}$ are perfect-memory context trees, then $\mathcal{A}\cup \mathcal{B}$ and $\mathcal{A}\cap \mathcal{B}$ are also perfect-memory context trees.
\end{property}

\begin{proof} Let $\mathcal{A}$ and $\mathcal{B}$ be perfect-memory context trees. Let $c \in (\mathcal{A} \cup \mathcal{B})^{*}$, and $i \in \{1, \dots, n\}$. We need to find a context $u \in (\mathcal{A}\cup\mathcal{B})^{*}$ such that $u\prec \overline{ca_{i}}$ . We separate into three cases.
 
 (i) If $c\in \mathcal{A}^*\cap \mathcal{B}^*$, $\exists a\in \mathcal{A}^*, b\in \mathcal{B}^*$, s.t. $a\prec \overline{ca_i}$ and $b\prec \overline{ca_i}$, because of $\mathcal{A}$ and $\mathcal{B}$ have perfect-memory and Definiton \ref{definition4}. $a$ and $b$ are both postfixes of $\overline{ca_i}$, if $b\prec a$, by Definition \ref{definition3}, $a\in (\mathcal{A}\cup \mathcal{B})^*$, done. If $a\prec b$, then $b\in (\mathcal{A}\cup \mathcal{B})^*$, done. 
 
 (ii) If $c\in \mathcal{A}^*\setminus \mathcal{B}^*$, $\exists b\in \mathcal{B}^*$, s.t. $b\prec c$ and $b\neq c$. $\exists a\in \mathcal{A}^*$, s.t. $a\prec \overline{ca_i}$ since $\mathcal{A}^*$ has perfect-memory. $\exists b'\in \mathcal{B}^*$, s.t. $b'\prec \overline{ba_i}$ since $\mathcal{B}^*$ has perfect-memory.
 Then we have $b'\prec\overline{ba_i}\prec \overline{ca_i}$ and $a\prec\overline{ca_i}$.
 If $a\prec b'$, then $b'\in (\mathcal{A}\cup \mathcal{B})^*$, done. If $b'\prec a$, then $a\in (\mathcal{A}\cup \mathcal{B})^*$, done.
 
 (iii) If $c\in \mathcal{B}^*\setminus \mathcal{A}^*$, proceed similarly to (ii).
 
 Therefore, $\mathcal{A}\cup \mathcal{B}$ is a perfect-memory context tree.
 Similarly, we can prove $\mathcal{A}\cap \mathcal{B}$ is a perfect-memory context tree.
 \end{proof}

 Since perfect-memory is closed under intersection, we can define the following.
\begin{definition}\label{definition5}
Let $\mathcal{T}$ be a context tree. The \textit{``perfect-memory closure of $\mathcal{T}$''}, denoted as $\mathcal{\overline{T}}$, is the intersection of all the perfect-memory context trees that contain $\mathcal{T}$, i.e. $$\overline{\mathcal{T}}=\bigcap_{\substack{\mathcal{T}\subseteq \mathcal{G}_i, \\ \mathcal{G}_i \in \mathcal{PM}}}{\mathcal{G}_i}$$
\end{definition}

Having perfect-memory is preserved under intersection in view of Property \ref{property2}, so the perfect-memory closure $\mathcal{\overline{T}}$ of a context tree $\mathcal{T}$ has perfect-memory. Thus, by definition, $\mathcal{\overline{T}}$ is the minimal context tree that contains $\mathcal{T}$ and has perfect-memory.

\begin{notation}
\label{notation1}
Let $C(\mathcal{T})$ denote the minimum complete context tree that contains $\mathcal{T}$. \\
\end{notation}
\begin{property}
$\overline{C(\mathcal{T})}=\overline{\mathcal{T}}$.
\end{property}
\begin{proof}
For any context tree $\mathcal{T}$, we have $\overline{C(\mathcal{T})}\subseteq \overline{C(\overline{\mathcal{T}}})=\overline{\overline{\mathcal{T}}}= \overline{\mathcal{T}}\subseteq \overline{C(\mathcal{T})}$ because of Property~\ref{property1}. Therefore, $\overline{C(\mathcal{T})}=\overline{\mathcal{T}}$.
\end{proof}

\begin{notation}
\label{notation2}
For a context tree $\mathcal{T}$ and an $i\in \{1,\dots,n \}$, by {\bf $\mathcal{T}_i$} we denote the subtree of $\mathcal{T}$ whose root is $a_{i}$. Note that $\mathcal{T}_{i}$ can be empty, and $\mathcal{T}_{i}^*=\{u\in S|\ \overline{ua_{i}}\in \mathcal{T}^*\}$. Moreover, let $\mathcal{T}_{\star}$ denote set of all subtrees of $\mathcal{T}$ that the root is a node of $\mathcal{T}$, i.e. $\mathcal{T}_{\star}=\{\mathcal{S}:\mathcal{S}=((\mathcal{T}_{i_1})_{i_2})...\}$. The figure below serves as an illustration. \\

\end{notation}
 \includegraphics[scale=0.9]{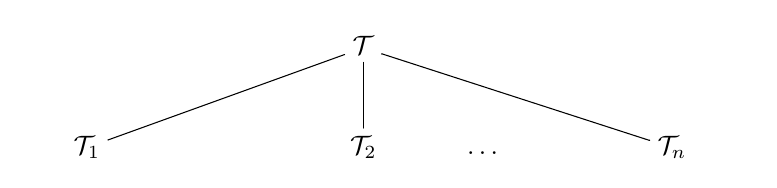}\\
 
Now we are ready to prove our first main result.

 \begin{proof}[Proof of Theorem\ref{theorem2}]
 Let $\mathcal{T}$ be a complete context tree. It follows from Lemma~\ref{lemma1} that when $\mathcal{T}_i\neq \emptyset,\ i\in \{1,...,n\}$, the following statements are equivalent.

 (i) $ \mathcal{T}_{i} \subseteq \mathcal{T}$.

 (ii) $\forall c' \in \mathcal{T}^*,\ \exists\  c\in \mathcal{T}_{i}^*$, s.t.  $c\prec c'$.
 
  \noindent Therefore, since
 $\emptyset \subseteq \mathcal{T}$, the statement of the theorem about the complete context tree $\mathcal{T}$ can be reformulated as follows: $\mathcal{T}$ has perfect-memory if and only if   $\ \forall i\in\{1,...,n\}$, either $\mathcal{T}_i=\emptyset$ or $\forall c' \in \mathcal{T}^*,\ \exists\  c\in \mathcal{T}_{i}^*$, s.t.  $c\prec c'$.
 
 {\bf Necessity. }Assume that $\mathcal{T}$ has perfect-memory. Let $i\in\{1,...,n\}$ and $c' \in \mathcal{T}^*$. By Definition \ref{definition4}, $\exists u\in \mathcal{T}^*$, s.t. $u\prec \overline{c'a_i}$. Let $c\in S$ so that $u=\overline{ca_i}$. If $a_i\in \mathcal{T}^*$, then $\mathcal{T}_i=\emptyset$. Otherwise, if $a_i \notin \mathcal{T}^*$, then c is not the empty string and $c\in \mathcal{T}_i^*$. Since $\overline{ca_i}=u\prec \overline{c'a_i}$, we have $c\prec c'$.
 
 {\bf Sufficiency.} Assume that $ \forall i\in \{1,...,n\}$, either $\mathcal{T}_i=\emptyset$ or $\forall c' \in \mathcal{T}^*,\ \exists\  c\in \mathcal{T}_{i}^*$, s.t.  $c\prec c'$. Let $i\in\{1,...,n\}$ and $c' \in \mathcal{T}^*$. We distinguish two cases, and show in both that there exists $u\in \mathcal{T}^*$ with $u\prec \overline{c'a_i}$. If $\mathcal{T}_i=\emptyset$, then $a_i\in \mathcal{T}^*$, and hence we can choose $u=a_i$. Otherwise, if $\mathcal{T}_i\neq \emptyset$, by assumption, $\exists c\in \mathcal{T}_i^*$, s.t. $c\prec c'$. In this case, we can choose $u=\overline{ca_i}$. Thus, $\mathcal{T}$ has perfect-memory.
 
%
%
%
\end{proof}
 
Now we discuss some corollaries of the characterization in Theorem~\ref{theorem2}.

\begin{notation}\label{notation3} 
If a context $c=\overline{c_ia_i}$, let $c_i$ denote the string of the context c without the last letter $a_i$.\\
\end{notation}
\begin{corollary}\label{corollary1}
A context tree $\mathcal{T}$ has perfect-memory if and only if the following two conditions hold: [1] $\mathcal{T}$ is complete, and [2] $\forall i\in \{1,...,n\},\ \forall c_i,\ s.t.\ \overline{c_ia_i}\in \mathcal{T}^*,\ \exists \ u \in \mathcal{T}^*\ s.t.\ c_i\prec u$.

\end{corollary}

\begin{proof} This corollary follows by Theorem~\ref{theorem2} and Lemma~\ref{lemma1}. \end{proof}

\begin{corollary}\label{corollary2}
A context tree $\mathcal{T}$ has perfect-memory if and only if the following two conditions hold: [1] $\mathcal{T}$ is complete, and [2] $\forall u$ for which $\exists \ w\in S, \ \overline{uw}\in \mathcal{T}^*$, we have $\exists \ w'\in S\  s.t.\ \overline{w'u} \in \mathcal{T}^*$.\end{corollary}

\begin{proof} The sufficiency follows by Corollary \ref{corollary1} just letting $w=a_i$.

 To prove the necessity now we assume $\mathcal{T}$ has perfect-memory, by Property~\ref{property1} we have $\mathcal{T}$ is complete. Take  $\forall u\ s.t.\ c=\overline{uw}\in \mathcal{T}^*$, 
when $w$ is the empty string, then it is trivial that when $w'$ is the empty string, $\overline{w'u} \in \mathcal{T}^*$. When $w$ is not the empty string, we assume the necessary requirement is not true,
 i.e. $\exists u_0\in \mathcal{T}^*$, s.t.  $u=\overline{w_0u_0}$ where $w_0$ is not the empty string, since $\mathcal{T^*}$ is complete. Let $w=\overline{a_{w(1)}a_{w(2)}\dots a_{w(l)}}$, then because $\mathcal{T}$ has perfect-memory, we have $\exists \ u_{k}\in \mathcal{T}^*$,  for $k\in \{1,...,l\}$, s.t. $u_k \prec \overline{u_{k-1}a_{w(k)}}$. Therefore, $u_l\prec \overline{w_0u_l}\prec \overline{w_0u_{l-1}a_{w(l)}}\prec \overline{w_0u_{l-2}a_{w(l-1)}a_{w(l)}}\prec \dots \prec \overline{w_0u_0a_{w(1)}a_{w(2)}\dots a_{w(l)}}=\overline{uw}=c$. Since $c$ and $u_l$ are both contexts in $\mathcal{T}^*$, and $\mathcal{T}$ is complete, $c=u_l$, this is a contradiction of $w_0$ is not the empty string. Hence, the necessary requirement is true.
 \end{proof}

\begin{corollary} \label{corollary3}
Let $\mathcal{T}$ be a complete context tree, then $\overline{\mathcal{T}}$ is the union at the root of all the subtrees of $\mathcal{T}$, i.e. $$\overline{\mathcal{T}}=\bigcup_{\mathcal{S}\in\mathcal{T_{\star}}}{\mathcal{S}} .$$
\end{corollary}
See Notation~\ref{notation2}.
\begin{proof} It is clear that any $\mathcal{S}$, $\mathcal{S}\in \mathcal{T}_{\star}$, is complete and $\bigcup_{\mathcal{S}\subseteq\mathcal{T}_{\star}}{\mathcal{S}} $ is complete.
$\forall c_0$ s.t. $\overline{c_0w_0} \in (\bigcup_{\mathcal{S}\subseteq\mathcal{T}_{\star}}{\mathcal{S}})^*$, then $\exists \mathcal{T}_{\alpha}\in\mathcal{T}_{\star}$ such that $\overline{c_0w_0}\in \mathcal{T}_{\alpha}^*$. Therefore, $\overline{c_0w_0w_{\alpha}}\in \mathcal{T}^*$,
then $\exists \mathcal{T}_{\beta}\in\mathcal{T}_{\star}$ such that $c_0\in \mathcal{T}_{\beta}^*$. By Definition~\ref{definition2}, $\exists t\in (\bigcup_{\mathcal{S}\in\mathcal{T_{\star}}}{\mathcal{S}})^*$ s.t. $c_0 \prec t$ since $\mathcal{T}_{\beta}\subseteq \bigcup_{\mathcal{S}\in\mathcal{T_{\star}}}{\mathcal{S}}$. By Corollary \ref{corollary2} we know that $\bigcup_{\mathcal{S}\in\mathcal{T_{\star}}}{\mathcal{S}}$ has perfect-memory. Since $\mathcal{T}\subseteq \bigcup_{\mathcal{S}\in\mathcal{T_{\star}}}{\mathcal{S}} $, we have $\overline{\mathcal{T}}\subseteq \overline{\bigcup_{\mathcal{S}\in\mathcal{T_{\star}}}{\mathcal{S}} }=\bigcup_{\mathcal{S}\in\mathcal{T_{\star}}}{\mathcal{S}}$.

On the other hand, $\forall c \in  (\bigcup_{\mathcal{S}\in\mathcal{T_{\star}}}{\mathcal{S}})^*$, $\exists \ \mathcal{T}_{\gamma}\in\mathcal{T}_{\star}$, s.t. $c\in \mathcal{T}_{\gamma}^*$, thus $\overline{cw_{\gamma}} \in \mathcal{T}^*$. Assume there is no $c'\in \overline{\mathcal{T}}^*$ s.t. $c\prec c'$. Then $\exists u\neq \emptyset$, s.t. $c=\overline{uc'}$ and $c' \in \overline{\mathcal{T}}^*$, since $\overline{\mathcal{T}}$ is complete. We have $\overline{uc'w_{\gamma}}\in \mathcal{T}^*$, $\exists w_1$ s.t. $\overline{w_1uc'w_{\gamma}}\in \overline{\mathcal{T}}^*$, since $\mathcal{T}\subseteq \overline{\mathcal{T}}$. $\overline{\mathcal{T}}$ has perfect-memory, by Corollary~\ref{corollary2}, $\exists w_2\in S$ s.t. $\overline{w_2w_1uc'}\in \overline{\mathcal{T}}^*$. This is a contradiction of  $c' \in \overline{\mathcal{T}}^*$ and $u\neq\emptyset$. Therefore, $\exists c'\in \overline{\mathcal{T}}^*$ s.t. $c\prec c'$. By Lemma~\ref{lemma1}, we have  $\bigcup_{\mathcal{S}\in\mathcal{T_{\star}}}\mathcal{S}\subseteq \overline{\mathcal{T}}$. Therefore, $\overline{\mathcal{T}}=\bigcup_{\mathcal{S}\in\mathcal{T_{\star}}}{\mathcal{S}} $.

\end{proof}

  \section{A Linear Algorithm For finding Perfect-memory Closure}\label{sec:alg} 
  
  \begin{notation}
\label{notation4}
Given a context tree $\mathcal{T}$, we have the following notations:\\
\begin{tabular}{r p{10cm}}
 $\ell$ & the depth of the context tree $\mathcal{T}$. \\
 $n(\mathcal{T})$ & the number of leaves in the $\mathcal{T}^*$.\\
 \end{tabular}\\
 \end{notation}

Given a complete context tree $\mathcal{T}$, we can find $\overline{\mathcal{T}}$ using Corollary~\ref{corollary3}. But it is the most redundant procedure unless we do it in parallel. The running time is $ O(\ell^2\cdot n( \mathcal{T}))$. 
We give a fast algorithm to calcuate $\overline{\mathcal{T}}$ without using parallel computing. It achieves the perfect-memory closure by extending the leaves of $\mathcal{T}$ using the equivalent conditions in Corollary~\ref{corollary1}. \\

\noindent {\bf Trimming Algorithm.} . \\
{\bf Input:} A complete context tree $\mathcal{T}$ from an alphabet $A=\{a_{1}, \, \dots, \, a_{n}\}$.\\
{\bf Output:} $\overline{\mathcal{T}}$, the perfect-memory closure of $\mathcal{T}$.\\
{\bf Ideas:} \\
1. Applying the equivalent conditions in Corollary 1 to enlarge $\mathcal{T}$ by adding the corresponding leaves or branches. \\
2. Visiting nodes in a depth decreasing order. Remove the branch after visited, terminated when there are no more unvisted nodes. \\
3. In the beginning of this algorithm we trim the complete context into its parent tree and in end we saturate the context tree to a complete tree by grow children for each node on the alphabet $A$.\\
{\bf Initialization:} Set $\mathcal{T}_p $ be the parent tree of $\mathcal{T}$. Set $\overline{\mathcal{T}}=\emptyset$.\\
{\bf Iteration:} \\
Step 1: If the length of $\mathcal{T}_p$ is less or equal to 1, go to step 4. Otherwise go to step 2.\\
Step 2: Let c be a deepest node in $\mathcal{T}.$ Let $\mathcal{P}_c^- $ denote the path $\mathcal{P}_c$ without the last node. If {$\mathcal{P}_c^- \notin \mathcal{T}_p$}, then
     
      $\ \ \ \ \ \mathcal{T}_p=\mathcal{T}_p\cup \mathcal{P}_c^-$;
      
       $\ \ \ \ \ \overline{\mathcal{T}}=\overline{\mathcal{T}}\cup \mathcal{P}_c^-$.\\     
Step 3: Remove $\mathcal{P}_c$ from $\mathcal{T}_p$ and go to step 1.\\
Step 4. $\overline{\mathcal{T}}:\leftarrow$ $\overline{\mathcal{T}}$ union the children of $\overline{\mathcal{T}}$ on the alphabet $A$.

See Appendix \ref{alg}.\\

\begin{property}\label{propertyruntime}The running time for the Trimming Algorithm is $O(\ell \cdot n(\overline{\mathcal{T}}))$. 
\end{property}

\begin{proof}

We use a standard data structure for trees: each node maintains a list of pointers to its children and a pointer to its parent. In addition, an auxiliary array can be implemented to have quick access to nodes with certain depth: for example, an array, each slot of which contains a link to a double-linked list containing nodes having a depth indicated by the index of that slot. It takes $O(n(\mathcal{T}))$ to initialize the structure.

With this data structure, during each iteration, Step 1 is $O(1)$. In Step 2 and 3, finding $c$ is $O(1)$ and finding $\mathcal{P}_c^-$ takes $O(\ell)$. Both operations union and delete can be done by following the path from root to $c$ and modify pointers along the way, which means they take $O(\ell)$ time. There are $n(\overline{\mathcal{T}})$ iterations, one for visiting each node in $\overline{\mathcal{T}}$.

After the loop, Step 4 takes $O(n(\overline{\mathcal{T}}))$.

Therefore, the running time of the Trimming Algorithm is $O(\ell \cdot n(\overline{\mathcal{T}}))$.

\end{proof}

\begin{definition}\label{definiton7} 
 A complete leaf set is a set of any $|A|$ leaves who have a same parent.
\end{definition}

\begin{property}\label{property5}
 If $\mathcal{A}$, $\mathcal{B}$ are perfect-memory context trees and $\mathcal{A}\supset \mathcal{B}$, there exists a sequence of perfect-memory conext trees $\mathcal{T}_1, \mathcal{T}_2, \ldots, \mathcal{T}_k$ such that $\mathcal{A}\supset  \mathcal{T}_1 \ \ldots \supset \mathcal{T}_k \supset \mathcal{B}$, and the difference between each two successive trees is exactly one complete leaf set. 
\end{property}
\begin{proof}
Let $\mathcal{T}_1$ be the tree obtained from $\mathcal{A}$ by cutting off the deepest complete leaf set that is not in $\mathcal{B}$. Similarly let $\mathcal{T}_2$ be the tree obtained from $\mathcal{T}_1$ by cutting off the deepest complete leaf set that is not in $\mathcal{B}$. Repeat this procedure until we get the sequence end with $\mathcal{B}$. Since we always delete the deepest leaf set, by Corollary~\ref{corollary1}, it is easy to see that all the new trees are also perfect-memory.
\end{proof}
Note that we can pick any deepest leaves set to remove, there might be many different sequences that satisfy the conditions in Property~\ref{property5}. 

%
%
  \begin{notation}
\label{notation5}
Given a complete context tree $\mathcal{T}$, we have the following notations:\\
\begin{tabular}{r p{10cm}}
 $n$ & size of the alphabet.\\
 $\mathcal{T}^\ell$ & the context tree with all the leaves of length $\ell$.\\
 $v(\mathcal{\mathcal{T}})$ & the number of nodes in the context tree $\mathcal{T}$.\\
 $r_1$ & $n(\mathcal{T})/n(\mathcal{T}^\ell)$\\
 $r_2$ & $n(\overline{\mathcal{T}})/n(\mathcal{T})$ (When $\mathcal{T}$ is perfect-memory $r_2=1$)\\ 
 \end{tabular}\\
 \end{notation}
  Therefore, we have
 $v(\mathcal{T}^\ell)=\frac{n(n^\ell-1)}{n-1}$ and 
 $n(\mathcal{T}^\ell)=n^\ell$.
  We consider rate $r_1$ as a measurement of how sparse a context tree is and rate $r_2$ as a measurement of how far a context tree is to perfect-memory.
\begin{property} \label{property6}
$1\leq r_2 \leq \ell$.
\end{property}
\begin{proof}
Assume we have a complete context tree $\mathcal{T}$. Let $\mathcal{T}^*=\{c_1,\ c_2,\ \dots,\ c_m\}$, where $c_i$ are different contexts in $\mathcal{T}^*$. Thus, $n(\mathcal{T})=m$. We say that a string $v\in S$ is a \textit{prefix} of a string $s\in S$, denoted by $v\succ s$, if there exists a string $w\in S$ such that $s=\overline{vw}.$ Let $\mathcal{C}_i^*:=\{u:u\succ c_i\}$. By Corollary \ref{corollary3} and Definition \ref{definition3}, we have $\overline{\mathcal{T}}^*=\big(\bigcup\mathcal{C}_i\big)^*\subseteq\bigcup_{i=1}^m{\mathcal{C}_i^*} .$ Therefore,
$$1\leq r_2=\frac{|\overline{\mathcal{T}}^*|}{|\mathcal{T}^*|}\leq \frac{|\mathcal{C}_1^*|+|\mathcal{C}_2^*|+\cdots+|\mathcal{C}_m^*|}{m}\leq \max_{1\leq i\leq m} |\mathcal{C}_i^*|=\ell.$$
\end{proof}

\section{Examples}\label{sec:examples}

\noindent {\bf{Example 1: Comb Tree}}

  \Tree[.Root [.0 [.0 
                       [.0 [.0 ] [.1 ]] 
                       [.1 ] ]
                       [.1 ]
                        ]
                 [.1 ]
         ]

This is a perfect-memory the context tree with the minimum rate $r_1$.\\ 
$$r_1 = \frac{(n-1)\ell+1}{n^\ell}\ \to \ 0,\ \text{ as } \ell \to \infty$$


\maketitle

\noindent {\bf{Example 2: A Sparse Tree}}\\

\Tree [.Root [.0 [.0 0 [.1 0 1 !{\brOverrideBoth}  ] !{\brOverrideBoth} ] [.1 [.0 [.0 0 [.1 0 1 !{\brRestore} ] ] 1 ] 1 !{\brRestore} ] !{\brRestore} ] [.1 [.0 [.0 0 [.1 0 1 !{\brOverrideBoth} ] ] 1 ] 1 ] !{\brRestore} ]

$\mathcal{T}$ is the tree with solid lines. The solid lines and dashed lines all together made $\overline{\mathcal{T}}$. $\mathcal{T}$ is not perfect-memory. $\mathcal{T}$ has the same number of leaves as the binary comb tree (example 1). \\


\noindent {\bf{Example 3: The smallest compete context tree whose perfect-memory closure contains full m-MC. (m=4)}}\\
\begin{center}
\Tree [.Root [.0 [.0 [.0 0 1 ] [.1 0 1 ] ] [.1 [.0 [.0 0 1 ] [.1 0 1 ] ] [.1 [.0 0 1 ] [.1 0 1 ] ] ] ] [.1 [.0 [.0 0 1 !{\brOverrideBoth} ] [.1 0 1 ] ] [.1 [.0 0 1 ] [.1 0 1 ] ] ] !{\brRestore} ]
\end{center}


Same as the previous example, the original trees are in solid lines. The dashed lines are the leaves added in the perfect-memory closure. Since for each new leaf of depth m is added if and only if there is another leaf of depth $m+1$ in the original tree by Corollary \ref{corollary1},
 we find the context tree in the figure above has the smallest number of leaves and the perfect-memory closure contains the full m-MC. For binary case we have:\\
$\ r_1 = \frac{3\times 2^{\ell-3}+1}{2^\ell}\ \to \frac{3}{8},\ \text{as}\   \ell\to \infty$\\
$\ \ \ \ \ \ \ \quad\hspace{4.5em} r_2=\frac{5\times 2^{\ell-3}}{3\times 2^{\ell-3}+1}\ \to \ \frac{5}{3},\ \text{as}\  \ell\to \infty$\\

\noindent {\bf{Example 4: A complete context tree with a large $r_2$} }\\

\begin{tikzpicture}[x=0.5cm, y=0.5cm,
    emph/.style={edge from parent/.style={draw, dashed}}]
\node at (0, 0) (a) {Root}
	child {node at (-2, -1) {$a_1$}
		child {node at (-4.5, -1) {$a_2$}
			child {node at (-3, -1) {$a_1$}}
			child {node at (-2, -1) (f) {$a_2$}}
			child  {node at (3, -1) (g) {$a_n$}}}}
	child [emph] {node at (-1, -1)  {$a_2$}
		child {node at (-1, -1)  {$a_1$}}
		child {node at (-0.5, -1) (d) {$a_2$}}
		child {node at (5, -1) (e) {$a_n$}}}
	child [emph] {node at (-0.4, -1) (b) {$a_3$}}
	child [emph] {node at (5, -1) (c) {$a_n$}};

\path (b) -- (c) node[midway] {$\cdots$};
\path (d) -- (e) node[midway] {$\cdots$};
\path (f) -- (g) node[midway] {$\cdots$};
\end{tikzpicture}

\begin{center}
Parent Tree of $\mathcal{T}$ and $\overline{\mathcal{T}}$
\end{center}
The figure above is an example for a large $r_2$. The parent tree of $\mathcal{T}$ is in solid lines and the parent tree of its perfect-memory closure is in both solid and dashed lines. In the figure above $\ell=4$, in general we have $n(\mathcal{T})=n^2+(\ell-2)(n-1)$,  $n(\overline{\mathcal{T}})=n^2(\ell-1)+n(2-\ell)+(n-1)(\ell-2)(\ell-3)/2.$\\
Thus, $r_2\to\ell-1,\ \text{as}\   n\to \infty$. Noting that $n$, the size of the alphabet, is finite.

%
%

%
%
\section{Acknowledgment}
I am very grateful to Mikhail Malioutov, who introduced me to SCOT and gave me this problem to solve and Gabor Lippner, who guided me with the research on the perfect-memory context tree.

\section*{Appendix}

\alglanguage{pseudocode}

\begin{algorithm}[H]
\caption{Trimming Algorithm.}
\label{alg}
\textbf{Input:} A complete context tree $\mathcal{T}$ from an alphabet $A=\{a_{1}, \, \dots, \, a_{n}\}$.\\
\textbf{Output:} $\overline{\mathcal{T}}$, the perfect-memory closure of $\mathcal{T}$.
\begin{algorithmic}[1]
\State $\mathcal{T}_p:\leftarrow$ the parent tree of $\mathcal{T}$.
\State $\overline{\mathcal{T}}:\leftarrow \emptyset$
\State If the length of $\mathcal{T}_p$ is no more than 1, go to line 8. Otherwise go to line 4.  
\State Pick a deepest node c in $\mathcal{T}_p$.    
    \If {$\mathcal{P}_c^- \nsubseteq \mathcal{T}_p$}
     
      $\mathcal{T}_p=\mathcal{T}_p\cup \mathcal{P}_c^-$;
      
       $\overline{\mathcal{T}}=\overline{\mathcal{T}}\cup \mathcal{P}_c^-$;
       
   \EndIf 
   
\State Remove $\mathcal{P}_c$ from $\mathcal{T}_p$ and go to line 3.
\State $\overline{\mathcal{T}}:\leftarrow$ $\overline{\mathcal{T}}$ union the children of $\overline{\mathcal{T}}$ on the alphabet $A$.
\end{algorithmic}
\end{algorithm}


\end{document}